\documentclass[runningheads]{llncs}

\usepackage[T1]{fontenc}
\usepackage{amsmath,amssymb,mathtools}
\usepackage{microtype}         
\usepackage{graphicx}
\usepackage{algorithm}
\usepackage{algpseudocode}

\begin{document}

\title{Navigating Small-World Networks with Distance Predictions }

\author{Ladan Kian \and Ming Ming Tan \and Dariusz Kowalski}

\authorrunning{L. Kian et al.}

\institute{Augusta University, Augusta, GA, USA\\
\email{lkian@augusta.edu}\\
\email{mtan@augusta.edu}\\
\email{dkowalski@augusta.edu}}

\maketitle
\begin{abstract}
The small-world phenomenon was given an algorithmic foundation by Kleinberg, who showed that in an augmented $k$-dimensional lattice a decentralized greedy algorithm delivers a message in $O(\log^2 n)$ expected steps. We study predicted-greedy routing, in which a mobile agent forwarding the message moves at each step to the neighbor minimizing a noisy $(\varepsilon,\delta)$-prediction of its distance to the target, redrawn at every step from an oracle conditioned on the full routing history. Two cases arise from what this agent can observe. An agent with the coordinate awareness can still compute lattice distance exactly, but not graph distance in the shortcut-augmented network, since that depends on the shortcuts of nodes it has not yet visited; given an $(\varepsilon,\delta)$-prediction of graph distance, information the classical model never supplies, it achieves expected delivery time $O(\log n/(1-4k\varepsilon\delta))$, an asymptotic improvement over $\Theta(\log^2 n)$. An agent with no coordinate awareness at all, the natural model for a privacy-preserving network whose nodes never disclose their coordinates, cannot compute even lattice distance; given an $(\varepsilon,\delta)$-prediction of lattice distance instead, it still reaches the target in $O(n/(1-4k\varepsilon\delta))$ expected steps.
Together these results show that a modest amount of predicted information, of the right kind, is enough to accelerate decentralized routing well below Kleinberg's classical bound, and that even when nodes reveal no coordinates at all, reliable delivery remains achievable.


\end{abstract}

\section{Introduction}
\label{sec:intro}
The small-world phenomenon, the empirical observation that arbitrary pairs of individuals in large social networks are connected by surprisingly short chains of acquaintances, was first documented by Milgram's letter-forwarding experiments~\cite{traversmilgram1969} and given a network-theoretic foundation by Watts and Strogatz~\cite{wattsstrogatz1998}, who showed that
superimposing a small number of long-range random links, shortcuts, onto a highly
clustered lattice produces graphs with both high clustering and small
diameter. Kleinberg~\cite{kleinberg2000nature,Kleinberg2000STOC} observed that low diameter alone does not explain Milgram's experiment: what is striking is that individuals, using only local information, are collectively able to find these short paths. Kleinberg formalized this as a decentralized routing problem and introduced the augmented-lattice model that underlies our topology, in which each node of a
$k$-dimensional grid keeps its local lattice links and additionally draws a single long-range shortcut with probability proportional to $r(u,w)^{-k}$, the lattice distance raised to the $-k$ power. Kleinberg showed that a simple decentralized greedy algorithm, always advancing to the neighbor of smallest known distance to the target, attains delivery time
$O(\log^2 n)$ precisely when the shortcut exponent is tuned to the lattice dimension, and that no decentralized algorithm can do so for any other exponent~\cite{kleinberg2000nature,kleinberg2001nips}. Martel and Nguyen~\cite{NguyenMartel2004PODC} subsequently showed this $O(\log^2 n)$ bound is in fact tight, $\Theta(\log^2 n)$, and that the graph's expected diameter is a full $\log n$ factor smaller, $\Theta(\log n)$: a message exists along an $O(\log n)$-length path, but no decentralized greedy crawler operating on exact local distances alone can find it, only a path
$\log n$ times longer. 

Throughout this classical literature, the distance information available at each node, whether lattice distance to a query target or the identity of a node's own shortcut, is assumed to be known exactly. This is a strong and, for many of the systems that motivate small-world models, an unrealistic assumption. A decentralized system that routes social-network queries,
peer-to-peer lookups, or mobile-agent search rarely has access to ground truth: what it has is a prediction, learned from prior traffic, cached from an earlier crawl, inferred from partial signals, or supplied by an external oracle, none of which is guaranteed correct. This motivates
asking a question the classical analysis of Kleinberg's model does not answer: how much does a small-world network's navigability degrade when the information driving greedy routing is imperfect, and whether the degradation can be controlled as a clean
function of the estimate's error. We answer this question for two natural
notions of imperfect distance information available to a decentralized crawler.

We answer this question inside the algorithms with predictions
paradigm~\cite{mitzenmachervassilvitskii2022}, in which a classical algorithm is
augmented with an untrusted, error-parameterized oracle and analyzed as a function of the oracle's error, interpolating between worst-case and oracle-optimal performance. Concretely, we equip the crawler with a
prediction oracle that, at every step, returns for each candidate neighbor an estimate of its distance to the target, guaranteed only to be correct in order with probability at least $1-\varepsilon$ between any two
neighbors, and correct in magnitude up to an additive error $\delta$.
Crucially, the oracle is queried fresh at every step, conditioned on the crawler's entire history so far, including any previous visit to the same node: a stronger, more honest requirement than assuming the predictions are fixed once before routing begins, and one that turns out to be exactly what is needed to keep the analysis well behaved.

A navigation query, in our setting as in Kleinberg's, is executed by a single mobile agent, the crawler, that moves one hop at a time and is not trusted to write to, or persistently remember information about, the nodes it visits: only to read a bounded amount of locally available information at its current position. What the crawler can read, however, is not fixed by the topology alone, and this is where our two cases diverge. If the
crawler has full coordinate awareness, it can already compute the exact lattice distance from any neighbor to the target; the only genuinely unknown quantity is the graph distance in the shortcut-augmented network, since that depends on shortcuts belonging to nodes the crawler has not yet visited. A prediction is therefore only informative here if it predicts graph distance (Case~1). If instead the crawler has no coordinate awareness at all, neither its own coordinate nor the target's, then even lattice distance is unknown, and a
prediction of lattice distance becomes the meaningful primitive (Case~2).
This second setting models a privacy-preserving network in which a node's position reflects unobserved latent attributes rather than disclosed coordinates, in the spirit of hidden-metric models of navigable networks~\cite{bogunapapadopouloskrioukov2010}. Hiding
coordinates from the crawler is a privacy guarantee for the nodes, not merely a modeling restriction; correspondingly, an untrusted
oracle providing predicted lattice distance without ever exposing the underlying coordinates is a natural way to preserve that guarantee while still enabling navigation.

\subsection{Our Contribution}
\paragraph{Significance.} We give the first analysis, to our knowledge, of
decentralized greedy routing in a Kleinberg small-world network under the algorithms-with-predictions paradigm, in which an untrusted, error-parameterized oracle replaces exact distance information at every routing step. We study the two prediction primitives motivated above:
predictions of graph distance under full coordinate awareness (Case~1), and predictions of lattice distance under no coordinate awareness (Case~2). For
Case~1, predicted-greedy routing achieves expected delivery time
$O(\log n / (1-4k\varepsilon\delta))$: for any fixed error parameters with
$4k\varepsilon\delta<1$, this is an asymptotic improvement over Kleinberg's own $\Theta(\log^2 n)$ greedy bound, not merely a recovery of it in the noiseless limit, since even an imperfect graph-distance oracle carries more routing-relevant information than exact lattice distance alone: its predictions incorporate information about the effect of the unobserved shortcuts on shortest-path distances in the realized network, without necessarily identifying those shortcuts explicitly. 
As $\varepsilon,\delta \to 0$ the bound degrades gracefully to $O(\log n)$, matching the graph's true expected diameter established by Martel and Nguyen~\cite{NguyenMartel2004PODC} up to constants, so no accuracy is free-lunched: the bound simply shows that a good-enough oracle lets a decentralized crawler approach the performance of an agent with global knowledge of the network.

For Case~2, we identify and resolve a robustness question: whether prediction error can cause the crawler to cycle indefinitely once it breaks the strict per-step progress that exact greedy routing enjoys for free. We show this cannot happen: because the oracle reissues a fresh, history-conditioned prediction at every step, plain greedy routing terminates in finite expected time, giving a delivery time bound linear rather than polylogarithmic in $n$. Section \ref{sec:delivery-analysis} explains precisely why our argument certifies only this weaker rate, and we return to whether a sharper analysis is possible in Section~\ref{sec:conclusion}.

\paragraph{Key ideas.} Both results follow the same template: identify a one-step ``excess distance'' random variable measuring how far a single routing decision falls short of the best available neighbor, bound its tail
probability using the oracle's monotonicity guarantee and its magnitude using the oracle's additive-accuracy guarantee, and telescope the resulting one-step drift into a bound on the expected hitting time via a direct additive-drift argument, without appealing to any independence assumption across steps that a revisit would violate. The two cases diverge exactly
where the underlying metric diverges: graph distance is itself a
shortest-path metric on the full network, so no neighbor, including the shortcut, can ever undercut the graph-optimal neighbor, giving a two-sided bound on the excess distance. Lattice distance is not a shortest-path metric once shortcuts are added, so the shortcut can outperform the best
local lattice neighbor by an arbitrary amount, giving only a one-sided bound; this asymmetry is the precise reason Case~2's rate is weaker, the analysis never needs, and therefore never benefits from, the shortcut's
power-law placement that drives Kleinberg's original speedup.

\section{Related Work}
\label{sec:related-work}
Our decentralized routing problem sits at the intersection of two lines of work: the classical analysis of Kleinberg's augmented-lattice small-world model, and the recent algorithms-with-predictions paradigm. We discuss each
in turn, then position our contribution against both.

\paragraph{Classical small-world navigability.} 
Kleinberg's original result fixed the shortcut exponent to the lattice dimension and showed greedy routing achieves $O(\log^2 n)$ expected
delivery~\cite{kleinberg2000nature,Kleinberg2000STOC}, later shown tight, $\Theta(\log^2 n)$, by Martel and Nguyen~\cite{NguyenMartel2004PODC}, who
also established the graph's expected diameter as the strictly smaller
$\Theta(\log n)$, the gap our Case~1 result closes under a sufficiently accurate oracle. Martel and Nguyen further characterize small-world graphs
more broadly, including geometries beyond the basic grid~\cite{NguyenMartel2005SODA},
and Kleinberg's own survey consolidates the model and its decentralized search guarantees~\cite{Kleinberg2006ICM}. The diameter of the underlying shortcut percolation structure itself has been pinned down precisely as a function of the shortcut exponent~\cite{CoppersmithGamarnikSviridenko2002},
and later work gives exact asymptotic characterizations of navigability at
and around Kleinberg's critical exponent~\cite{CarettaCartozoDLR09,CarmiCSB09}.
Most recently, Alimohammadi et al.~\cite{AlimohammadiIsikSaberi2025}
use local weak convergence to characterize the Kleinberg model's structure
in the large-$n$ limit and pin down the same critical exponent from a different, local-limit perspective, corroborating why navigability is so
exponent-sensitive. All of this literature analyzes a single canonical
process, greedy hops toward a target under exact distances; a
different, complementary process on the same topology, the unbiased random walk, exhibits its own, distinctly located phase transition in mixing time as a function of the shortcut exponent~\cite{DyerGalanisGoldbergJerrumVigoda2020},
underscoring that the exponent-navigability relationship is a property of the network, not an artifact of the particular decentralized process studied.
It has to be noted that small-world networks were not only used for point-to-point navigability, but also for more complex communication and computation tasks, e.g., computing symmetric functions via aggregation, see e.g.,~\cite{KamenevKM23}.

\paragraph{Improving greedy with additional information.} A
substantial body of work improves on plain greedy by giving nodes more, but still perfectly reliable, local knowledge. Manku, Naor and
Wieder~\cite{MankuNaorWieder2004STOC} and Naor and
Wieder~\cite{NaorWieder2004IPTPS} show that lookahead to a neighbor's
neighbors reduces delivery time to $O(\log^{1+1/k} n)$ in expectation;
Fraigniaud et al.~\cite{FraigniaudGavoillePaul2004PODC} obtain a
similar improvement via topological awareness of nearby shortcuts. Zeng et al.~\cite{zenghsuwang2005} augment each node with $O(\log n)$ bits of local awareness of nearby shortcuts to reach a near-optimal
$O(\log n \log\log n)$ expected hop count, and Ruas~\cite{Ruas2013Report}
shows that combining a power-law contact distribution with one-hop
lookahead can drive the expected number of hops down to $O(1)$ in a
one-dimensional variant of the model. This line of work asks how much more exact information a decentralized node needs to beat plain
greedy, the opposite direction from ours: we ask how much routing degrades when the node's information, even of the same scope as Kleinberg's, is no
longer exact. Giakkoupis and Schabanel~\cite{GiakkoupisSchabanel2011STOC}
and Fraigniaud and Giakkoupis~\cite{FraigniaudGiakkoupis2010STOC} sharpen
the exponent-versus-dimension transition and extend navigability to broader
underlying graph families; Fraigniaud et al.~\cite{fraigniaudlebharlotker2006} generalize navigability to
arbitrary doubling-dimension graphs. 

\paragraph{Geographic and geometric routing under imperfect location
information.} A separate, older literature studies greedy forwarding in
physically embedded geometric networks, and is the closest prior work,
outside the predictions paradigm, to asking what happens when a
crawler's position or distance information is not perfectly reliable.
Karp and Kung's GPSR~\cite{karpkung2000} and Kranakis, Singh and Urrutia's
compass routing~\cite{kranakissinghurrutia1999} both forward greedily by
geographic or angular proximity to the destination, falling back to a
perimeter- or face-traversal recovery rule when greedy forwarding reaches a
local void, under the standing assumption that every node's coordinates are
exact. Seada, Helmy and Govindan~\cite{seadahelmygovindan2004} study
exactly what happens when that assumption fails: they show that even small
localization errors, as little as ten percent of radio range, can cause
geographic face routing to fail non-recoverably, and quantify the resulting
performance degradation empirically. This is, to our knowledge, the closest
existing result in spirit to ours, imperfect position information degrading
a greedy geometric routing rule, but the setting, guarantee, and error
model all differ substantially from ours: it concerns a fixed, per-instance
geometric embedding and empirically measured failure rates under
deterministic perimeter recovery, rather than an $(\varepsilon,\delta)$-accurate
stochastic oracle re-queried at every step within Kleinberg's shortcut-augmented
lattice, and it does not provide an expected-delivery-time bound as a closed-form function of the error, which is the central object of our analysis.

\paragraph{Hidden-metric and privacy-motivated models.} Our Case~2
crawler, which observes node identifiers but no coordinates, is motivated by hidden-metric models of real navigable networks, in which a node's position reflects latent, unobserved attributes rather than disclosed coordinates~\cite{bogunapapadopouloskrioukov2010}. In that literature the hidden geometry is
recovered statistically from the observed graph rather than supplied by an untrusted oracle at query time; our contribution is to ask what decentralized navigation looks like when the crawler is given only a noisy, per-query prediction of this hidden distance and nothing else, never the coordinates themselves, which is the privacy-preserving reading of the model motivated in Section~1.

\paragraph{Algorithms with predictions.} The algorithms-with-predictions
paradigm augments a classical algorithm with an untrusted,
error-parameterized oracle and analyzes performance as a function of the oracle's error, interpolating between worst-case and oracle-optimal
behavior~\cite{mitzenmachervassilvitskii2022}. It has been applied to online
caching~\cite{lykouriasvassilvitskii2021}, non-clairvoyant and online
scheduling~\cite{purohitsvitkinakumar2018}, and contract
scheduling~\cite{angelopouloskamali2023}, and has only very recently entered
distributed and self-stabilizing computing. Boyar, Ellen and
Larsen~\cite{boyarellenlarsen2025} initiate deterministic distributed graph
algorithms with predictions in synchronous message passing, illustrated on
Maximal Independent Set, where every node receives a prediction of its own
output bit and predictions are evaluated against the same round-complexity
measure used without predictions. Aradhya and
Scheideler~\cite{aradhyascheideler2025} study self-stabilizing graph
linearization under untrusted advice, and Balliu et
al.~\cite{balliuetal2025} study distributed computation with local advice
more generally. None of these consider navigation: in all three,
every node receives one prediction about a static, global property of the
graph or its own final output, evaluated once, whereas in our setting a
single mobile crawler receives a fresh, history-conditioned prediction at
every step of a dynamic query, about a value, distance to a
query-specific target, that is not fixed in advance and depends on which
node has been asked. This distinction is exactly what let us dispense with
the memory and cycle-avoidance machinery that a fixed, once-issued
prediction would have required (Section~4): the dynamic, per-step,
history-conditioned nature of the oracle is not an incidental modeling
choice but the mechanism that makes plain greedy routing terminate with no
auxiliary state at all. To the best of our knowledge, no prior work in the
predictions literature considers decentralized navigation in a small-world
topology, the setting of this paper.

\section{Model and Preliminaries}
\label{sec:model}

\subsection{Network Model}
\label{sec:network-model}
Fix an integer $k\geq 1$, the dimension, and an integer $n\geq 2$, the
side length. The node set is the $k$-dimensional grid
$V=\{1,\ldots,n\}^{k}$, $N=|V|=n^{k}$; each $u\in V$ is a lattice
point $u=(u_1,\ldots,u_k)$, and each node is additionally assigned a
unique identifier, independent of its coordinate. Let
$r(u,v)=\sum_{j=1}^{k}|u_j-v_j|$
denote the \textbf{lattice distance}, following Kleinberg's original
grid formulation \cite{Kleinberg2000STOC,kleinberg2000nature}.

Each node $u$ has an undirected local link to every node at lattice
distance $1$; write $N(u)$ for this set of \emph{lattice neighbors}.
Interior nodes satisfy $|N(u)|=2k$; nodes within distance $1$ of the
grid boundary have fewer neighbors. There are $O(n^{k-1}) =
O(N^{(k-1)/k})$ such boundary-affected nodes, a vanishing fraction of
$N$, and we follow the standard convention of treating this as a
negligible edge effect rather than modeling it explicitly
\cite{NguyenMartel2004PODC}; statements below about "a node $u$" hold
for interior nodes, which suffices with high probability for a
uniformly random source or target.

In addition, each node $u$ independently creates one directed \emph{shortcut link} $(u,W_u)$, whose endpoint $W_u\ne u$ is sampled
with probability
\[
    \Pr[W_u = w] \;=\; \frac{r(u,w)^{-k}}{Z_u},
    \qquad
    Z_u \;=\! \sum_{x\in V\setminus\{u\}} r(u,x)^{-k}.
\]

These choices are independent across nodes and links. All
shortcut links are sampled once before routing begins and remain
fixed throughout the routing process. There exist constants $C_{\mathrm{lo}},C_{\mathrm{up}}>0$, depending
only on $k$, with $C_{\mathrm{lo}}\log n \le Z_u \le C_{\mathrm{up}}
\log n$ for every interior $u\in V$ and all sufficiently large $n$
\cite{Kleinberg2000STOC,NguyenMartel2004PODC}.

This is the $p=q=1$, with the clustering exponent $\alpha=k$ specialization of
Kleinberg's construction \cite{Kleinberg2000STOC}, generalized to
$k$ dimensions as in \cite{NguyenMartel2004PODC}; $\alpha=k$ is the
unique exponent for which a decentralized algorithm can achieve
polylogarithmic expected delivery time \cite{Kleinberg2000STOC}, and
under it, lattice greedy routing (forwarding, at each step, to the
neighbor of smallest true lattice distance to the target) achieves
expected delivery time $\Theta(\log^2 n)$, tight, with the graph's
expected diameter a further $\log n$ factor smaller, $\Theta(\log n)$
\cite{NguyenMartel2004PODC}; expectations throughout are over the
random shortcuts and over a source and target drawn uniformly at
random~from~$V$.

Let $G=(V,E)$ be the resulting fixed directed graph, each undirected
local adjacency represented by two opposite directed links, together
with every node's shortcut link. For $u\in V$, write $N^{+}(u):=N(u)\cup\{W_u\}$ for its set of outgoing neighbors.

We use two distinct notions of distance on $G$ throughout the paper.
The \emph{lattice distance} $r(u,v)$, defined above, ignores all
shortcuts and depends only on the underlying grid; its maximum value
is $L=O(n)=O(N^{1/k})$. The \textbf{graph distance}
\[
    \ell_G(u,v) \;=\; \text{length of a shortest directed path from $u$ to $v$ in $G$}
\]
is the shortest-path distance in the full network, shortcuts
included, and is finite for every $u,v\in V$ since the local lattice
edges alone connect the grid. Since every lattice path is in
particular a path in $G$, $\ell_G(u,v)\le r(u,v)\le L$ for all
$u,v\in V$; the two notions coincide only when no shortcut shortens
the route. Section~\ref{sec:case1} studies predictions of $\ell_G$,
and Section~\ref{sec:case2} studies predictions of $r$.

A \textbf{navigation query} is an ordered pair $(s,t)\in V\times V$, the source and target, executed by a single mobile agent, the
\textbf{crawler}, initially located at $s$. When located at a node
$u\ne t$, a \emph{routing step} consists of selecting an outgoing
neighbor $v\in N^{+}(u)$ and traversing the edge $(u,v)$; we call the
selection of $v$ the \emph{routing decision} of that step. The query
is delivered the first time the crawler reaches $t$, and the
\emph{delivery time} $T$ is the number of routing steps taken.

What the crawler observes at each node, in particular, whether
coordinates are visible, and whether the available predictions
concern graph distance or lattice distance, differs between the two
settings studied in this paper, and is specified separately in
Section~\ref{sec:case1} and Section~\ref{sec:case2}.

\subsection{Case 1: Graph-Distance Predictions}
\label{sec:case1}

In this setting the crawler has full coordinate awareness: it is
given the coordinate of the target $t$, and, when located at a node
$u$, it knows the coordinate of $u$ and the identities and
coordinates of every node in $N^{+}(u)$, including the endpoint of
its shortcut link. It can therefore compute the exact lattice
distance $r(v,t)$ for every $v\in N^{+}(u)$. However, it is not given the shortcut links of nodes it has not yet visited; since these unobserved links can shorten paths to the target, the crawler does not in general know $\ell_G(v,t)$ for $v\in N^{+}(u)$. 
For the fixed target $t$ of a given navigation query, write
\[
    h_t(u) := \ell_G(u,t) \qquad \text{for every } u\in V,
\]
the true shortest directed graph distance from $u$ to $t$ in the
realized network $G$.

\textbf{Prediction oracle.} We adopt the algorithms with predictions
paradigm \cite{mitzenmachervassilvitskii2022}; the oracle returns an error-parameterized
prediction vector at each step, formalized as follows.

Index routing steps $i=0,1,\ldots$, let $U_i$ denote the crawler's
location at the start of step $i$ ($U_0=s$), and let $\mathcal{F}_i$
denote the pre-step history, the source, the target, all nodes
visited up to and including $U_i$, all previously issued predictions,
and all previous routing decisions, excluding the prediction issued
at step $i$ itself.

\begin{definition}[$(\varepsilon,\delta)$-graph-distance prediction oracle]
\label{def:oracle-case1}
Fix $ 0 \leq \varepsilon \leq 1$ and $0 \leq \delta$. A prediction oracle is an
$(\varepsilon,\delta)$-graph-distance prediction oracle if, for every
realized network $G$, every target $t$, and every feasible pre-step
history $\mathcal{F}_i$ with $U_i=u\ne t$, the oracle returns a
prediction vector $\hat h_i=(\hat h_i(v):v\in N^{+}(u))$ satisfying:
\begin{enumerate}
\item[(i)] \emph{Conditional pairwise monotonicity.} For every
$v,w\in N^{+}(u)$ with $h_t(v)<h_t(w)$,
\[
    \Pr\bigl[\hat h_i(w) \le \hat h_i(v) \mid \mathcal{F}_i\bigr] \le \varepsilon.
\]
\item[(ii)] \emph{Additive approximation.} Conditional on
$\mathcal{F}_i$, with probability $1$,
\[
        \left|\hat{h}_i(v)-h_t(v)\right|
        \leq \delta
        \qquad
        \text{for every }v\in N^{+}(u).
\]
\end{enumerate}
\end{definition}

A new prediction vector is generated at every routing step, including a
step at which $U_i$ is a node the crawler has already visited. The
oracle may choose the distribution of $\hat h_i$ depending on the
complete history $\mathcal{F}_i$; in particular, prediction vectors
generated at different routing steps need not be independent.

\begin{definition}[Predicted-greedy routing]
\label{def:predicted-greedy-case1}
At every routing step $i$ with $U_i\ne t$, the crawler receives the
prediction vector $\hat h_i$ and selects an outgoing neighbor that
minimizes the predicted graph distance to the target,
\[
    U_{i+1} \in \arg\min_{v\in N^{+}(U_i)} \hat h_i(v),
\]
with a fixed, arbitrary tie-breaking rule when the minimum is
attained by more than one neighbor.
\end{definition}

Algorithm~\ref{alg:predicted-greedy} in Appendix~\ref{sec:alg} gives pseudocode for this update rule, covering both Case~1 and Case~2 (Definition~\ref{def:predicted-greedy-case2}).

\subsection{Case 2: Lattice-Distance Predictions}
\label{sec:case2}
In this setting the crawler has no coordinate awareness whatsoever.
Nodes are distinguished only by identifier, and, by construction
(Section~\ref{sec:network-model}), an identifier reveals nothing
about a node's coordinate. When located at a node $u$, the crawler
observes the identifier of $u$, the identifiers of the nodes in
$N^{+}(u)$, which of them is the shortcut neighbor, and, for every
$v\in N^{+}(u)$, a prediction of the lattice distance $r(v,t)$; it
observes neither its own coordinate, nor $t$'s, nor that of any
neighbor. This restriction is essential to the model, not incidental
to it: were coordinates observable, as in Case 1, the
crawler could compute $r(v,t)$ exactly for every $v\in N^{+}(u)$
without consulting any oracle, and a lattice-distance prediction
would carry no information. 

Hiding coordinates is therefore what
makes lattice-distance prediction a meaningful primitive here, and it
is also the more realistic assumption for a social network, in which
a node's position corresponds to unobserved latent attributes
(interests, communities, or other features) rather than to disclosed
physical coordinates, in the spirit of hidden-metric models of
navigable networks \cite{bogunapapadopouloskrioukov2010}.
For the fixed target $t$ of a given navigation query, we write
\[
r_t(u):=r(u,t) \qquad \text{for every} u\in V,
\]
the true lattice distance from $u$ to the target $t$.

\begin{definition}[$(\varepsilon,\delta)$-lattice-distance prediction oracle]
\label{def:oracle-case2}
Fix $0\le\varepsilon\le 1$ and $0 \leq \delta$. A prediction oracle is
an $(\varepsilon,\delta)$-lattice-distance prediction oracle if, for
every realized network $G$, every target $t$, and every feasible
pre-step history $\mathcal{F}_i$ with $U_i=u\ne t$, the oracle
returns a prediction vector $\hat p_i = (\hat p_i(v) : v\in
N^{+}(u))$ satisfying:
\begin{enumerate}
\item[(i)] \emph{Conditional pairwise monotonicity.} For every
$v,w\in N^{+}(u)$ with $r_t(v)<r_t(w)$,
\[
    \Pr\bigl[\hat p_i(w) \le \hat p_i(v) \mid \mathcal{F}_i\bigr] \le \varepsilon.
\]
\item[(ii)] \emph{Additive approximation.} Conditional on
$\mathcal{F}_i$, with probability $1$,
\[
    \bigl|\hat p_i(v) - r_t(v)\bigr| \le \delta
    \qquad \text{for every } v \in N^{+}(u).
\]
\end{enumerate}
\end{definition}

As in Case 1, a new prediction vector is generated at every routing step, including
a step at which $U_i$ is a node the crawler has already visited. The
oracle may choose the distribution of $\hat p_i$ depending on the
complete history $\mathcal{F}_i$; in particular, prediction vectors
generated at different routing steps need not be independent.

\begin{definition}[Predicted-greedy routing, lattice-distance case]
\label{def:predicted-greedy-case2}
At every routing step $i$ with $U_i \neq t$, the crawler receives the
prediction vector $\hat p_i$ and selects an outgoing neighbor that minimizes
the predicted lattice distance to the target,
\[
U_{i+1} \in \arg\min_{v \in N^+(U_i)} \hat p_i(v),
\]
with the same fixed, arbitrary tie-breaking rule as in Definition~\ref{def:predicted-greedy-case1} when the
minimum is attained by more than one neighbor.
\end{definition}

Unlike the graph distance $h_t$ used in Case~1, the lattice distance $r_t$ is not itself a shortest-path metric on $G$: a shortcut endpoint
$w = W_u$ can satisfy $r_t(w) \ll r_t(u) - 1$, so the true minimum
$\min_{v \in N^+(u)} r_t(v)$ may fall strictly below $r_t(u) - 1$, whereas at least one lattice neighbor attains $r_t(u)-1$ exactly, possibly several when $u$ and $t$ differ in more than one coordinate, and no lattice neighbor can do better than $r_t(u) - 1$.
This asymmetry between the local neighbors, which change
$r_t$ by exactly one step, and the shortcut, which may change it by an
arbitrary amount, is the source of the different delivery-time argument required for this case.


\subsection{Model Comparison and Oracle Information}\label{sec:discussion}
Both settings of Case $1$ and Case $2$ (see subsections~\ref{sec:case1},~\ref{sec:case2}) use the same $k$-dimensional Kleinberg topology, with nearest-neighbor lattice links and one independently sampled directed shortcut per node, fixed throughout the routing \cite{Kleinberg2000STOC,NguyenMartel2004PODC}. The distinction is the information available to the crawler. Classical lattice-greedy routing in \cite{kleinberg2000nature,Kleinberg2000STOC} uses coordinates to compute \emph{exact} lattice distances. Case $1$ retains this coordinate access and additionally supplies predictions of the graph distance $\ell_G(v,t)$, whereas Case $2$ hides coordinates and supplies predictions of the lattice distance $r(v,t)$. Both cases use the same memoryless rule of choosing a neighbor of minimum predicted distance. 
Thus, Case $1$ provides information beyond the classical model, while Case $2$ studies imperfect access to the lattice distances used by classical greedy routing. 
In particular, the graph-distance prediction oracle (Definition~\ref{def:oracle-case1}) must incorporate information about the influence of shortcuts in  the realized network, including those belonging to unvisited nodes; it is not sufficient to know only their sampling distribution. 
Indeed, consider two shortcut realizations $G$ and $G'$ that are compatible with the same complete observation history of the crawler, including its current local observations.
For the same queried neighbor $v$ and target $t$, suppose that
$|\ell_G(v,t)-\ell_{G'}(v,t)|>2\delta$.
An oracle whose output distribution is determined solely by that
observation history must use the same distribution in both
realizations. However, the intervals
$[\ell_G(v,t)-\delta,\ell_G(v,t)+\delta]$ and
$[\ell_{G'}(v,t)-\delta,\ell_{G'}(v,t)+\delta]$
are disjoint, so that distribution cannot satisfy the almost-sure
additive guarantee in both realizations. In such situations,
knowledge of the shortcut-sampling distribution and the crawler's
observation history is insufficient: the oracle requires additional
instance-specific information about the effects of unobserved
shortcuts, without necessarily identifying those shortcuts
explicitly.
In contrast, the lattice-distance prediction oracle does not require information about unobserved shortcuts since $r(v,t)$ depends only on the underlying lattice embedding. 
Hence, the faster delivery time in Case $1$ should be interpreted as a consequence of a stronger information model, with our analysis quantifying the dependence
of the delivery-time bound on prediction error. 
In both cases, the oracle guarantees must hold after every feasible routing history, including revisits. The construction of such an oracle and the cost of acquiring its information is outside the scope of our analysis.

\subsection{Delivery Time}
\label{sec:delivery-time}

The delivery time is the random variable
\begin{equation}
    T := \min\{i \ge 0 : U_i = t\},
    \label{eq:delivery-time}
\end{equation}
with the convention that $T = \infty$ if the crawler never reaches the
target. Since $U_0 = s$ and each routing step traverses exactly one link,
$T$ is the number of routing steps, equivalently the number of hops, taken
to reach $t$. This definition applies uniformly to both cases studied
below: it depends only on the routing-step process $(U_i)_{i \ge 0}$ fixed
in Section~\ref{sec:network-model}, not on which distance the crawler's
predictions concern.

For a fixed graph $G$ and fixed endpoints $s,t \in V$, we write
$\mathbb{E}[T \mid G,s,t]$ for the expectation of $T$ taken over the
prediction oracle's random outputs alone. Whenever we additionally average
over the random shortcut links or over the endpoints $s,t$, we state this
explicitly.

\section{Expected Delivery Time}
\label{sec:delivery-analysis}
We analyze the delivery time of predicted-greedy routing separately for each of the two prediction models introduced in Section~\ref{sec:model}.

\subsection{Case 1: Graph-Distance Predictions}
\label{sec:analysis-case1}

\textbf{One-Step Progress.}
Fix a realization of $G$, a source $s$, a target $t$, and a routing step
$i < T$. Let
\[
H_i := h_t(U_i)
\]
denote the true graph distance remaining at the beginning of step $i$.
Since $U_i \neq t$, a shortest directed path from $U_i$ to $t$ begins with
an outgoing neighbor. Choose any such neighbor $v_i^* \in N^+(U_i)$. Then
$h_t(v_i^*) = H_i - 1$.
The choice of $v_i^*$ is used only in the analysis; the crawler does not know $v_i^*$. Define the one-step excess distance by
\begin{equation}
    R_i := h_t(U_{i+1}) - (H_i - 1).
    \label{eq:Ri-def}
\end{equation}
Thus, $R_i$ measures the additional graph distance incurred by the selected neighbor relative to an optimal next hop. In particular, $R_i = 0$ when the
selected link is the first link of a shortest directed path to $t$.

\begin{lemma}
\label{lem:one-step-case1}
For every fixed graph $G$, target $t$, and feasible pre-step history
$\mathcal{F}_i$ ending at $U_i \neq t$,
\begin{align}
    &\Pr[R_i>0\mid \mathcal{F}_i]
    \leq 2k\varepsilon,
    \label{eq:positive-excess-probability}\\
    & \Pr[0\leq R_i\leq2\delta\mid \mathcal{F}_i]=1,
    \label{eq:excess-range}\\
    &\mathbb{E}[R_i\mid \mathcal{F}_i]
    \leq 4k\varepsilon\delta.
    \label{eq:expected-excess}
\end{align}
\end{lemma}

\begin{proof}
For every outgoing neighbor $w \in N^+(U_i)$, the link $(U_i,w)$ followed by a shortest directed path from $w$ to $t$ gives a directed path from
$U_i$ to $t$. Therefore, $H_i \leq 1 + h_t(w)$, and hence
$h_t(w) \geq H_i - 1$. Applying this inequality to the selected neighbor $U_{i+1}$ shows that $R_i \geq 0$.

We first bound the probability that $R_i$ is positive. If $R_i > 0$, then predicted-greedy has selected an outgoing neighbor $w = U_{i+1}$ satisfying
$h_t(v_i^*) < h_t(w)$. Since $U_{i+1}$ minimizes $\hat h_i$ over
$N^+(U_i)$, and $v_i^* \in N^+(U_i)$, this selection necessarily satisfies
$\hat h_i(w) \le \hat h_i(v_i^*)$. Hence
$$
\{R_i > 0\} \subseteq \bigcup_{\substack{w \in N^+(U_i) \\ h_t(v_i^*) < h_t(w)}}
\big\{\hat h_i(w) \le \hat h_i(v_i^*)\big\}.
$$
The node $U_i$ has at most $2k$ local outgoing neighbors and one shortcut
outgoing neighbor, so $|N^+(U_i)| \le 2k+1$, and there are at most $2k$
possible choices of $w$ other than $v_i^*$. By
Definition~\ref{def:oracle-case1}(i), conditional pairwise monotonicity,
each such event satisfies
$\Pr[\hat h_i(w) \le \hat h_i(v_i^*) \mid \mathcal{F}_i] \le \varepsilon$.
Taking a union bound over these at most $2k$ events gives
\[
\Pr[R_i > 0 \mid \mathcal{F}_i] \le 2k\varepsilon.
\]

We already established $R_i \ge 0$ above. It remains to show $R_i \leq 2\delta$ almost surely conditional
on $\mathcal{F}_i$.
Since $U_{i+1}$ minimizes $\hat h_i$ over $N^+(U_i)$ by
Definition~\ref{def:predicted-greedy-case1}, and $v_i^* \in N^+(U_i)$, the
inequality
$\hat h_i(U_{i+1}) \le \hat h_i(v_i^*)$
holds surely, it is an immediate consequence of $U_{i+1}$ being an argmin,
true for every realization of $\hat h_i$. By the additive approximation guarantee,
Definition~\ref{def:oracle-case1}(ii), which holds with probability $1$
conditional on $\mathcal{F}_i$ for every $v \in N^+(U_i)$, applied first to
$U_{i+1}$ and then to $v_i^*$,
$$
h_t(U_{i+1}) \le \hat h_i(U_{i+1}) + \delta
\le \hat h_i(v_i^*) + \delta
\le h_t(v_i^*) + 2\delta
= H_i - 1 + 2\delta.
$$
By~\eqref{eq:Ri-def}, this is exactly $R_i \le 2\delta$, and, as a
combination of a sure fact (the argmin inequality) and an almost-sure fact
(the additive approximation), it holds almost surely conditional on
$\mathcal{F}_i$. Together with $R_i \ge 0$, established above, this shows
\[
\Pr[0 \le R_i \le 2\delta \mid \mathcal{F}_i] = 1.
\]

Since $R_i \ge 0$ always, the complement of $\{R_i > 0\}$ is exactly
$\{R_i = 0\}$; that is, $R_i = R_i \cdot \mathbf{1}\{R_i > 0\}$ surely.
Taking conditional expectations and using the almost-sure bound
$R_i \le 2\delta$ from~\eqref{eq:excess-range} on the event $\{R_i>0\}$,
$$
\mathbb{E}[R_i \mid \mathcal{F}_i]
= \mathbb{E}\big[R_i \cdot \mathbf{1}\{R_i > 0\} \mid \mathcal{F}_i\big]
\le 2\delta \, \Pr[R_i > 0 \mid \mathcal{F}_i].
$$
By~\eqref{eq:positive-excess-probability}, $\Pr[R_i > 0 \mid \mathcal{F}_i] \le 2k\varepsilon$, then
\[
\mathbb{E}[R_i \mid \mathcal{F}_i] \le 2\delta \cdot 2k\varepsilon = 4k\varepsilon\delta.
\]
\qed
\end{proof}

\noindent
\textbf{Drift Toward the Target.}
We now prove the following theorem.

\begin{theorem}
\label{thm:case1}
Fix a graph $G$, a source $s$, and a target $t$. Suppose that the
prediction oracle satisfies Definition~\ref{def:oracle-case1} after every
feasible pre-step history. If $4k\varepsilon\delta < 1$, then the expected
delivery time of predicted-greedy routing satisfies
\begin{equation}
    \mathbb{E}[T \mid G,s,t] \le \frac{h_t(s)}{1 - 4k\varepsilon\delta},
    \label{eq:thm1-fixed-graph}
\end{equation}
where the expectation is over the oracle's random outputs. Consequently,
if $G$ is drawn from the $k$-dimensional Kleinberg model of
Section~\ref{sec:network-model} with $p=q=1$ and $\alpha=k$, then, for
every fixed pair $s,t \in V$,
\begin{equation}
    \mathbb{E}[T \mid s,t] = O\!\left(\frac{\log n}{1-4k\varepsilon\delta}\right),
    \label{eq:thm1-averaged}
\end{equation}
where the expectation in~\eqref{eq:thm1-averaged} is over both the random
shortcut links and the oracle's random outputs. The hidden constant
depends only on $k$.
\end{theorem}

\begin{proof}
From the definition of $R_i$ in~\eqref{eq:Ri-def}, for every step $i < T$,
\[
H_{i+1} = H_i - 1 + R_i
\ .
\]
Taking conditional expectations and applying Lemma~\ref{lem:one-step-case1},
\begin{equation}
    \mathbb{E}[H_i - H_{i+1} \mid \mathcal{F}_i]
    = 1 - \mathbb{E}[R_i \mid \mathcal{F}_i]
    \ge 1 - 4k\varepsilon\delta.
    \label{eq:drift-case1}
\end{equation}
The assumption $4k\varepsilon\delta < 1$ ensures that the right-hand side of
\eqref{eq:drift-case1} is positive, so $H_i$ drifts downward in expectation
at every step before $T$.

We turn this drift into a bound on $\mathbb{E}[T]$ via a direct
additive-drift argument. Fix an integer $m \ge 1$ and consider the stopped
process at $\min\{T,m\}$. Multiplying~\eqref{eq:drift-case1} by the
indicator that step $i$ occurs before $\min\{T,m\}$, summing over
$i = 0,\ldots,m-1$, and taking expectations gives
$$
\mathbb{E}\big[H_0 - H_{\min\{T,m\}} \mid G,s,t\big]
\ge (1-4k\varepsilon\delta)\, \mathbb{E}[\min\{T,m\} \mid G,s,t].
$$
Since $H_{\min\{T,m\}} \ge 0$ and $H_0 = h_t(s)$, rearranging gives
$$
\mathbb{E}[\min\{T,m\} \mid G,s,t] \le \frac{h_t(s)}{1-4k\varepsilon\delta}.
$$
Letting $m \to \infty$ and applying monotone convergence to the
nondecreasing sequence $\min\{T,m\}$ yields~\eqref{eq:thm1-fixed-graph}.

To obtain~\eqref{eq:thm1-averaged}, we average~\eqref{eq:thm1-fixed-graph}
over the random Kleinberg graph. Martel and Nguyen showed that the expected
diameter of the $k$-dimensional Kleinberg graph with $p=q=1$ and exponent
$k$ is $\Theta(\log n)$ \cite{NguyenMartel2004PODC}. Since $h_t(s) = \ell_G(s,t) \le \max_{x,y \in V} \ell_G(x,y)$,
$$
\mathbb{E}[h_t(s)] \le \mathbb{E}\Big[\max_{x,y \in V} \ell_G(x,y)\Big] = O(\log n).
$$
Taking expectations in~\eqref{eq:thm1-fixed-graph} over the random graph
therefore gives 
\[
\mathbb{E}[T \mid s,t] = O\!\left(\frac{\log n}{1-4k\varepsilon\delta}\right).
\]
This completes the proof of Theorem \ref{thm:case1}.
\qed
\end{proof}

\paragraph{Revisiting nodes.}
The drift bound~\eqref{eq:drift-case1} holds conditionally after every
feasible pre-step history, including histories in which the crawler has
revisited one or more nodes: the argument never resamples an already
exposed shortcut link, nor treats a revisit as a fresh shortcut trial,
only the prediction vector $\hat h_i$ is regenerated at each step, and no independence between prediction vectors at different steps is required.
Since Theorem~\ref{thm:case1} gives $\mathbb{E}[T \mid G,s,t] < \infty$ for
every fixed $G,s,t$, it follows that $\Pr[T < \infty \mid G,s,t] = 1$;
the crawler reaches the target with probability one.

\subsection{Case 2: Lattice-Distance Predictions}
\label{sec:analysis-case2}

\textbf{One-Step Progress.}
Fix a realization of $G$, a source $s$, a target $t$, and a routing step
$i < T$. Let
\[
\rho_i := r_t(U_i)
\]
denote the true lattice distance remaining at the beginning of step $i$.
By the triangle inequality, every $v \in N(U_i)$ satisfies
$r_t(v) \ge r_t(U_i) - r(U_i,v) = \rho_i - 1$, and at least one lattice
neighbor attains this bound with equality (the neighbor moving one step
toward $t$ along any coordinate on which $U_i$ and $t$ differ). Choose any
such neighbor $v_i^* \in N(U_i)$. Then
$r_t(v_i^*) = \rho_i - 1$.
As in Case~1, the choice of $v_i^*$ is used only in the analysis. Define
the one-step excess distance by
\begin{equation}
    \Delta_i := r_t(U_{i+1}) - (\rho_i - 1).
    \label{eq:Delta-def}
\end{equation}
Unlike $R_i$ in Case~1, $\Delta_i$ is \emph{not} guaranteed nonnegative:
the shortcut endpoint $W_{U_i}$ may satisfy $r_t(W_{U_i}) < \rho_i - 1$,
since $r_t$ is not itself a shortest-path metric on $G$, so selecting the
shortcut can give $\Delta_i < 0$, unrequired extra progress. The lemma
below therefore bounds $\Delta_i$ only from above.

\begin{lemma}
\label{lem:one-step-case2}
For every fixed graph $G$, target $t$, and feasible pre-step history
$\mathcal{F}_i$ ending at $U_i \neq t$,
\begin{align}
    &\Pr[\Delta_i > 0 \mid \mathcal{F}_i] \le 2k\varepsilon,
    \label{eq:positive-excess-case2}\\
    &\Pr[\Delta_i \le 2\delta \mid \mathcal{F}_i] = 1,
    \label{eq:excess-upper-case2}\\
    &\mathbb{E}[\Delta_i \mid \mathcal{F}_i] \le 4k\varepsilon\delta.
    \label{eq:expected-excess-case2}
\end{align}
\end{lemma}
\begin{proof}
If $\Delta_i > 0$,
then predicted-greedy has selected an outgoing neighbor $w = U_{i+1}$
satisfying $r_t(v_i^*) < r_t(w)$. Since $U_{i+1}$ minimizes $\hat p_i$ over
$N^+(U_i)$, and $v_i^* \in N^+(U_i)$, this selection necessarily satisfies
$\hat p_i(w) \le \hat p_i(v_i^*)$. Hence
\[
\{\Delta_i > 0\} \subseteq \bigcup_{\substack{w \in N^+(U_i) \\ r_t(v_i^*) < r_t(w)}}
\big\{\hat p_i(w) \le \hat p_i(v_i^*)\big\}.
\]
There are at most $|N^+(U_i)| - 1 \le 2k$ choices of $w$ other than
$v_i^*$. By Definition~\ref{def:oracle-case2}(i), conditional pairwise
monotonicity, each such event satisfies
$\Pr[\hat p_i(w) \le \hat p_i(v_i^*) \mid \mathcal{F}_i] \le \varepsilon$.
A union bound over these at most $2k$ events gives
\[
\Pr[\Delta_i > 0 \mid \mathcal{F}_i] \le 2k\varepsilon.
\]
Since $U_{i+1}$
minimizes $\hat p_i$ over $N^+(U_i)$ by
Definition~\ref{def:predicted-greedy-case2}, and $v_i^* \in N^+(U_i)$, the
inequality $\hat p_i(U_{i+1}) \le \hat p_i(v_i^*)$ holds surely, for every
realization of $\hat p_i$. By the additive approximation guarantee,
Definition~\ref{def:oracle-case2}(ii), which holds with probability $1$
conditional on $\mathcal{F}_i$ for every $v \in N^+(U_i)$, applied first to
$U_{i+1}$ and then to $v_i^*$,
\[
r_t(U_{i+1}) \le \hat p_i(U_{i+1}) + \delta
\le \hat p_i(v_i^*) + \delta
\le r_t(v_i^*) + 2\delta
= \rho_i - 1 + 2\delta.
\]
By~\eqref{eq:Delta-def}, this is exactly $\Delta_i \le 2\delta$, and, as a
combination of a sure fact and an almost-sure fact, it holds almost surely
conditional on $\mathcal{F}_i$, proving
\[
\Pr[\Delta_i \le 2\delta \mid \mathcal{F}_i] = 1.
\]
Because $\Delta_i$ can be negative, write $\mathbb{E}[\Delta_i \mid \mathcal{F}_i]$
as the sum of its contributions from the two events $\{\Delta_i > 0\}$ and
$\{\Delta_i \le 0\}$:
\[
\mathbb{E}[\Delta_i \mid \mathcal{F}_i]
= \mathbb{E}\big[\Delta_i \cdot \mathbf{1}\{\Delta_i > 0\} \mid \mathcal{F}_i\big]
+ \mathbb{E}\big[\Delta_i \cdot \mathbf{1}\{\Delta_i \le 0\} \mid \mathcal{F}_i\big].
\]
The second term is at most $0$, since $\Delta_i \le 0$ on that event. The
first term is at most $2\delta \Pr[\Delta_i > 0 \mid \mathcal{F}_i]$ by
\eqref{eq:excess-upper-case2}, hence at most $4k\varepsilon\delta$ by
\eqref{eq:positive-excess-case2}. Adding the two bounds gives
\[
\mathbb{E}[\Delta_i \mid \mathcal{F}_i] \le 4k\varepsilon\delta.
\]
\qed
\end{proof}

\noindent
\textbf{Drift Toward the Target.}
We now prove the following theorem.

\begin{theorem}
\label{thm:case2}
Fix a graph $G$, a source $s$, and a target $t$. Suppose that the
prediction oracle satisfies Definition~\ref{def:oracle-case2} after every
feasible pre-step history. If $4k\varepsilon\delta < 1$, then the expected
delivery time of predicted-greedy routing satisfies
\begin{equation}
    \mathbb{E}[T \mid G,s,t] \le \frac{r(s,t)}{1 - 4k\varepsilon\delta},
    \label{eq:thm2-fixed-graph}
\end{equation}
where the expectation is over the oracle's random outputs. Consequently,
if $G$ is drawn from the $k$-dimensional Kleinberg model of
Section~\ref{sec:network-model} with $p=q=1$ and $\alpha=k$, then, for
every fixed pair $s,t \in V$,
\begin{equation}
    \mathbb{E}[T \mid s,t] = O\!\left(\frac{n}{1-4k\varepsilon\delta}\right)
    = O\!\left(\frac{N^{1/k}}{1-4k\varepsilon\delta}\right),
    \label{eq:thm2-averaged}
\end{equation}
where the expectation in~\eqref{eq:thm2-averaged} is over the random
source and target only. The hidden constant depends only on $k$.
\end{theorem}
\begin{proof}
From the definition of $\Delta_i$ in~\eqref{eq:Delta-def}, for every step
$i < T$,
\[
\rho_{i+1} = \rho_i - 1 + \Delta_i.
\]
Taking conditional expectations and applying
Lemma~\ref{lem:one-step-case2},
\begin{equation}
    \mathbb{E}[\rho_i - \rho_{i+1} \mid \mathcal{F}_i]
    = 1 - \mathbb{E}[\Delta_i \mid \mathcal{F}_i]
    \ge 1 - 4k\varepsilon\delta,
    \label{eq:drift-case2}
\end{equation}
positive under $4k\varepsilon\delta < 1$. Fix $m \ge 1$ and consider the
stopped process at $\min\{T,m\}$. Multiplying~\eqref{eq:drift-case2} by the
indicator that step $i$ occurs before $\min\{T,m\}$, summing over
$i=0,\ldots,m-1$, and taking expectations gives
\[
\mathbb{E}\big[\rho_0 - \rho_{\min\{T,m\}} \mid G,s,t\big]
\ge (1-4k\varepsilon\delta)\, \mathbb{E}[\min\{T,m\} \mid G,s,t].
\]
Since $\rho_{\min\{T,m\}} \ge 0$ and $\rho_0 = r_t(s) = r(s,t)$,
\[
\mathbb{E}[\min\{T,m\} \mid G,s,t] \le \frac{r(s,t)}{1-4k\varepsilon\delta}.
\]
Letting $m \to \infty$ and applying monotone convergence
proves~\eqref{eq:thm2-fixed-graph}.

To obtain~\eqref{eq:thm2-averaged}, average~\eqref{eq:thm2-fixed-graph} over
a uniformly random pair $s,t \in V$. Unlike $h_t(s)$ in Case~1, $r(s,t)$
does not depend on the random shortcuts, only on the grid geometry: for
$s,t$ drawn uniformly and independently, each coordinate gap
$|s_j-t_j|$ satisfies $\mathbb{E}|s_j-t_j| = \Theta(n)$, so
$\mathbb{E}[r(s,t)] = \sum_{j=1}^k \mathbb{E}|s_j-t_j| = \Theta(n)$. Taking
expectations in~\eqref{eq:thm2-fixed-graph} over $s,t$ gives
\[
\mathbb{E}[T \mid s,t] = O\!\left(\frac{n}{1-4k\varepsilon\delta}\right).
\]
\qed
\end{proof}

\paragraph{Revisiting nodes.}
As in Case~1, the drift bound established in the proof of Theorem~\ref{thm:case2} holds conditionally after every feasible pre-step history, including histories with repeated
visits: only the prediction vector $\hat p_i$ is regenerated at a
revisited node, no shortcut link is resampled, and no independence
across steps is required. Since Theorem~\ref{thm:case2} gives
$\mathbb{E}[T \mid G,s,t] < \infty$ for every fixed $G,s,t$, it follows that
$\Pr[T < \infty \mid G,s,t] = 1$.

A natural concern for lattice-distance predicted-greedy is that, unlike the
exact case $\delta=0$ where distance decreases strictly at every step, a
prediction error can cause the crawler to select a neighbor with
$\Delta_i > 0$, potentially revisiting a previously seen node. Theorem~\ref{thm:case2} shows that, because the
oracle reissues a fresh prediction vector $\hat p_i$ at every routing step,
including at a revisited node, conditioned on the complete history
$\mathcal{F}_i$ (Definition~\ref{def:oracle-case2}), the one-step guarantee
of Lemma~\ref{lem:one-step-case2} holds identically whether or not $U_i$
has been visited before. No bound on how many times a node may be
revisited is needed, and no shortcut link is ever resampled, only the
prediction is redrawn. The drift argument therefore gives a finite
expected delivery time, and hence almost-sure termination, using only the dynamic nature
of the oracle, with no auxiliary memory or cycle-avoidance rule imposed on the crawler.

\paragraph{Remark (a stabilization guarantee).}
Theorem~\ref{thm:case2} is best read as a stabilization guarantee for memoryless predicted-greedy routing under a dynamic, history-conditioned
prediction oracle. Prediction errors may still cause revisits, or even temporary cycles; what the dynamic oracle rules out is such an error
persisting indefinitely. Because a fresh prediction vector satisfying
Definition~\ref{def:oracle-case2} is issued after every feasible routing history, including immediately after a return to a previously visited
node, the positive conditional drift established in
Lemma~\ref{lem:one-step-case2} is restored at every revisit, exactly as if the crawler had never been there before. This gives a finite expected
delivery time and, hence, almost-sure termination, with no memory of past
visits. We emphasize how demanding this assumption is: at every step the oracle may choose a new prediction distribution depending on the entire
routing history so far, it may return a different prediction on a
repeated visit to the same node, and the guarantees of
Definition~\ref{def:oracle-case2} must hold after every feasible history, not just at the outset. Theorem~\ref{thm:case2} should be read as showing
what a sufficiently adaptive oracle can buy a memoryless crawler, not as a guarantee available from an arbitrary fixed predictor.

\paragraph{Remark (comparison with Case 1).}
The averaged bound~\eqref{eq:thm2-averaged} is linear in $n$, not polylogarithmic in $n$ like Case~1's~\eqref{eq:thm1-averaged}. The drift in Theorem~\ref{thm:case2} is generated entirely by the guaranteed per-step
decrease of $1$ from a local lattice neighbor, the same guarantee plain lattice-greedy has with no shortcuts at all; the argument never invokes
the shortcut's placement $\Pr[W_u=w] \propto r(u,w)^{-k}$, so Theorem~\ref{thm:case2} and its proof are unaffected if every shortcut link is deleted from $G$. This explains why this particular drift argument certifies only a linear rate: it simply does not use the one ingredient, the shortcut's power-law placement, that Kleinberg's original phase argument exploits to obtain $\Theta(\log^2 n)$. It does not show that no sharper analysis of the same algorithm could do better. Indeed, when
$\delta=0$ the oracle is exact and predicted-greedy coincides with plain lattice-distance greedy augmented with Kleinberg's shortcuts, whose
expected delivery time is $\Theta(\log^2 n)$~\cite{NguyenMartel2004PODC};
our bound~\eqref{eq:thm2-averaged} gives only $O(n)$ in this same limit.

\section{Conclusion}
\label{sec:conclusion}
We analyzed decentralized greedy routing in Kleinberg's small-world model under the algorithms with predictions paradigm, for two natural notions of
imperfect distance information. Under graph-distance predictions (Case~1), predicted-greedy routing improves asymptotically on Kleinberg's own
$\Theta(\log^2 n)$ bound, up to a clean, explicit degradation factor in the prediction error. Under lattice-distance predictions (Case~2), the more realistic model when coordinates are hidden, we established a stabilization guarantee: predicted-greedy routing provably terminates, with no memory of past visits, purely because of the dynamic, history-conditioned nature of the oracle. This gives a finite expected
delivery time, together with a precise account of why our argument
certifies a linear rather than polylogarithmic rate: the drift we exhibit uses only the guaranteed one-step progress of a local lattice neighbor and never the shortcut's power-law placement.

This leaves a natural open question: whether a stronger, polylogarithmic
bound is achievable for lattice-distance predictions under plain greedy routing without auxiliary memory, by exploiting the shortcut's power-law
placement directly rather than only its guaranteed local-neighbor
progress. Any such argument would need to control the independence of shortcut information across possible revisits without relying on an explicit cycle-avoidance mechanism, a genuinely different technical challenge from the one resolved here.

\begin{credits}

\subsubsection{\ackname}
Ladan Kian and Ming Ming Tan were supported in part by National Science Foundation (NSF) grant CCF-2348346. 

\subsubsection{\discintname}
The authors have no competing interests to declare that are relevant to the
content of this article.

\end{credits}

\bibliographystyle{splncs04}
\bibliography{References}

@article{kleinberg2000nature,
  author  = {Jon Kleinberg},
  title   = {Navigation in a Small World},
  journal = {Nature},
  volume  = {406},
  pages   = {845},
  year    = {2000},
  doi     = {10.1038/35022643}
}

@inproceedings{kleinberg2001nips,
  author    = {Jon Kleinberg},
  title     = {Small-World Phenomena and the Dynamics of Information},
  booktitle = {Advances in Neural Information Processing Systems 14 (NIPS)},
  publisher = {MIT Press},
  pages = {431--438},
  year      = {2001}
}

@article{wattsstrogatz1998,
  author  = {Duncan J. Watts and Steven H. Strogatz},
  title   = {Collective Dynamics of `Small-World' Networks},
  journal = {Nature},
  volume  = {393},
  pages   = {440--442},
  year    = {1998}
}

@article{traversmilgram1969,
  author  = {Jeffrey Travers and Stanley Milgram},
  title   = {An Experimental Study of the Small World Problem},
  journal = {Sociometry},
  volume  = {32},
  number  = {4},
  pages   = {425--443},
  year    = {1969}
}

@inproceedings{Kleinberg2000STOC,
  author    = {Jon Kleinberg},
  title     = {The Small-World Phenomenon: An Algorithmic Perspective},
  booktitle = {Proceedings of the 32nd Annual {ACM} Symposium on Theory of Computing (STOC'00)},
  pages     = {163--170},
  year      = {2000},
  publisher = {ACM},
  doi       = {10.1145/335305.335325}
}

@misc{AlimohammadiIsikSaberi2025,
  author       = {Yeganeh Alimohammadi and Senem I{\c{s}}{\i}k and Amin Saberi},
  title        = {Local Limits of Small World Networks},
  year         = {2025},
  howpublished = {arXiv:2501.11226}
}

@inproceedings{NguyenMartel2004PODC,
  author    = {Charles U. Martel and Van Nguyen},
  title     = {Analyzing Kleinberg's (and Other) Small-World Models},
  booktitle = {Proceedings of the Twenty-Third Annual {ACM} Symposium on Principles of Distributed Computing (PODC '04)},
  pages     = {179--188},
  year      = {2004},
  publisher = {ACM},
  doi       = {10.1145/1011767.1011794}
}

@inproceedings{NguyenMartel2005SODA,
  author    = {Van Nguyen and Charles U. Martel},
  title     = {Analyzing and Characterizing Small-World Graphs},
  booktitle = {Proceedings of the Sixteenth Annual ACM-SIAM Symposium on Discrete Algorithms (SODA'05)},
  pages     = {311--320},
  year      = {2005},
  publisher = {Society for Industrial and Applied Mathematics}
}

@article{KamenevKM23,
  author       = {Aleksandar Kamenev and
                  Dariusz R. Kowalski and
                  Miguel A. Mosteiro},
  title        = {Faster Supervised Average Consensus in Adversarial and Stochastic
                  Anonymous Dynamic Networks},
  journal      = {{ACM} Trans. Parallel Comput.},
  volume       = {10},
  number       = {2},
  pages        = {13:1--13:35},
  year         = {2023},
  url          = {https://doi.org/10.1145/3593426},
  doi          = {10.1145/3593426},
  bibsource    = {dblp computer science bibliography, https://dblp.org}
}

@inproceedings{FraigniaudGavoillePaul2004PODC,
  author    = {Pierre Fraigniaud and Cyril Gavoille and Christophe Paul},
  title     = {Eclecticism Shrinks Even Small Worlds},
  booktitle = {Proceedings of the 23rd Annual {ACM} Symposium on Principles of Distributed Computing (PODC '04)},
  pages     = {169--178},
  year      = {2004},
  publisher = {ACM},
  doi       = {10.1145/1011767.1011793}
}

@inproceedings{MankuNaorWieder2004STOC,
  author    = {Gurmeet Singh Manku and Moni Naor and Udi Wieder},
  title     = {Know Thy Neighbor's Neighbor: The Power of Lookahead in Randomized {P2P} Networks},
  booktitle = {Proceedings of the 36th Annual {ACM} Symposium on Theory of Computing (STOC '04)},
  pages     = {54--63},
  year      = {2004},
  publisher = {ACM},
  doi       = {10.1145/1007352.1007368}
}

@inproceedings{NaorWieder2004IPTPS,
  author    = {Moni Naor and Udi Wieder},
  title     = {Know Thy Neighbor's Neighbor: Better Routing for Skip-Graphs and Small Worlds},
  booktitle = {Peer-to-Peer Systems III (IPTPS 2004)},
  series    = {Lecture Notes in Computer Science},
  volume    = {3279},
  pages     = {269--277},
  publisher = {Springer},
  year      = {2005},
  doi       = {10.1007/978-3-540-30183-7\_26}
}

@article{CarettaCartozoDLR09,
  author  = {C{\'e}cile Caretta Cartozo and Paolo De Los Rios},
  title   = {Extended Navigability of Small World Networks: Exact Results and New Insights},
  journal = {Physical Review Letters},
  volume  = {102},
  number  = {23},
  pages   = {238703},
  year    = {2009},
  doi     = {10.1103/PhysRevLett.102.238703}
}

@article{CarmiCSB09,
  author  = {Shai Carmi and Stephen Carter and Jie Sun and Daniel ben-Avraham},
  title   = {Asymptotic Behavior of the Kleinberg Model},
  journal = {Physical Review Letters},
  volume  = {102},
  number  = {23},
  pages   = {238702},
  year    = {2009},
  doi     = {10.1103/PhysRevLett.102.238702},
  note    = {Companion paper to Caretta Cartozo and De Los Rios (2009), same issue}
}

@inproceedings{GiakkoupisSchabanel2011STOC,
  author    = {George Giakkoupis and Nicolas Schabanel},
  title     = {Optimal Path Search in Small Worlds: Dimension Matters},
  booktitle = {Proceedings of the 43rd Annual {ACM} Symposium on Theory of Computing (STOC '11)},
  pages     = {393--402},
  year      = {2011},
  publisher = {ACM},
  doi       = {10.1145/1993636.1993689}
}

@inproceedings{FraigniaudGiakkoupis2010STOC,
  author    = {Pierre Fraigniaud and George Giakkoupis},
  title     = {On the Searchability of Small-World Networks with Arbitrary Underlying Structure},
  booktitle = {Proceedings of the 42nd Annual {ACM} Symposium on Theory of Computing (STOC '10)},
  pages     = {389--398},
  year      = {2010},
  publisher = {ACM},
  doi       = {10.1145/1806689.1806744}
}

@inproceedings{fraigniaudlebharlotker2006,
  author    = {Pierre Fraigniaud and Emmanuelle Lebhar and Zvi Lotker},
  title     = {A Doubling Dimension Threshold $\Theta(\log\log n)$ for Augmented Graph Navigability},
  booktitle = {Proceedings of the 14th Annual European Symposium on Algorithms (ESA)},
  series    = {Lecture Notes in Computer Science},
  volume    = {4168},
  pages     = {376--386},
  publisher = {Springer},
  year      = {2006},
  doi       = {10.1007/11841036\_35}
}

@article{DyerGalanisGoldbergJerrumVigoda2020,
  author  = {Martin E. Dyer and Andreas Galanis and Leslie Ann Goldberg and Mark Jerrum and Eric Vigoda},
  title   = {Random Walks on Small World Networks},
  journal = {ACM Transactions on Algorithms},
  year    = {2020},
  volume  = {16},
  number  = {3},
  pages   = {37},
  doi     = {10.1145/3382208}
}

@inproceedings{zenghsuwang2005,
  author    = {Jianyang Zeng and Wen-Jing Hsu and Jiangdian Wang},
  title     = {Near Optimal Routing in a Small-World Network with Augmented Local Awareness},
  booktitle = {Parallel and Distributed Processing and Applications (ISPA 2005)},
  series    = {Lecture Notes in Computer Science},
  volume    = {3758},
  pages     = {503--513},
  publisher = {Springer},
  year      = {2005},
  doi       = {10.1007/11576235\_52}
}

@article{CoppersmithGamarnikSviridenko2002,
  author  = {Don Coppersmith and David Gamarnik and Maxim Sviridenko},
  title   = {The Diameter of a Long-Range Percolation Graph},
  journal = {Random Structures \& Algorithms},
  volume  = {21},
  number  = {1},
  pages   = {1--13},
  year    = {2002},
  doi     = {10.1002/rsa.10042}
}

@techreport{Ruas2013Report,
author      = {Olivier Ruas},
title       = {Neighbor-of-Neighbor Routing in Small-World Networks with Power-Law Degree},
institution = {HAL open archive, Distributed, Parallel, and Cluster Computing [cs.DC]},
year        = {2013},
note        = {HAL Id: dumas-00854954}
}

@inproceedings{Kleinberg2006ICM,
author    = {Jon Kleinberg},
title     = {Complex Networks and Decentralized Search Algorithms},
booktitle = {Proceedings of the International Congress of Mathematicians (ICM), Madrid 2006},
volume    = {3},
pages     = {1019--1044},
year      = {2006},
doi       = {10.4171/022-3/50},
note      = {Nevanlinna Prize lecture; proceedings volume published 2007}
}

@inproceedings{karpkung2000,
  author    = {Brad Karp and H. T. Kung},
  title     = {{GPSR}: Greedy Perimeter Stateless Routing for Wireless Networks},
  booktitle = {Proceedings of the 6th Annual International Conference on Mobile Computing and Networking (MobiCom)},
  pages     = {243--254},
  year      = {2000}
}

@inproceedings{kranakissinghurrutia1999,
  author    = {Evangelos Kranakis and Harvinder Singh and Jorge Urrutia},
  title     = {Compass Routing on Geometric Networks},
  booktitle = {Proceedings of the 11th Canadian Conference on Computational Geometry (CCCG)},
  pages     = {51--54},
  year      = {1999}
}

@inproceedings{seadahelmygovindan2004,
  author    = {Kamal Seada and Ahmed Helmy and Ramesh Govindan},
  title     = {On the Effect of Localization Errors on Geographic Face Routing in Sensor Networks},
  booktitle = {Proceedings of the 3rd International Symposium on Information Processing in Sensor Networks (IPSN)},
  pages     = {71--80},
  year      = {2004}
}

@article{bogunapapadopouloskrioukov2010,
  author  = {Mari{\'a}n Bogu{\~n}{\'a} and Fragkiskos Papadopoulos and Dmitri Krioukov},
  title   = {Sustaining the Internet with Hyperbolic Mapping},
  journal = {Nature Communications},
  volume  = {1},
  pages   = {62},
  year    = {2010}
}

@inproceedings{purohitsvitkinakumar2018,
  author    = {Manish Purohit and Zoya Svitkina and Ravi Kumar},
  title     = {Improving Online Algorithms via {ML} Predictions},
  booktitle = {Advances in Neural Information Processing Systems 31 (NeurIPS)},
  pages     = {9684--9693},
  year      = {2018}
}

@article{lykouriasvassilvitskii2021,
  author  = {Thodoris Lykouris and Sergei Vassilvitskii},
  title   = {Competitive Caching with Machine Learned Advice},
  journal = {Journal of the ACM},
  volume  = {68},
  number  = {4},
  pages   = {24:1--24:25},
  year    = {2021},
  doi     = {10.1145/3447579}
}

@article{mitzenmachervassilvitskii2022,
  author  = {Michael Mitzenmacher and Sergei Vassilvitskii},
  title   = {Algorithms with Predictions},
  journal = {Communications of the ACM},
  volume  = {65},
  number  = {7},
  pages   = {33--35},
  year    = {2022}
}

@article{angelopouloskamali2023,
  author  = {Spyros Angelopoulos and Shahin Kamali},
  title   = {Contract Scheduling with Predictions},
  journal = {Journal of Artificial Intelligence Research},
  volume  = {77},
  pages   = {395--426},
  year    = {2023},
  doi     = {10.1613/jair.1.14117}
}

@inproceedings{boyarellenlarsen2025,
  author    = {Joan Boyar and Faith Ellen and Kim S. Larsen},
  title     = {Brief Announcement: Distributed Graph Algorithms with Predictions},
  booktitle = {Proceedings of the {ACM} Symposium on Principles of Distributed Computing (PODC)},
  pages     = {322--325},
  year      = {2025},
  doi       = {10.1145/3732772.3733530},
  note      = {Full version: arXiv:2501.05267}
}

@misc{aradhyascheideler2025,
  author       = {Vijeth Aradhya and Christian Scheideler},
  title        = {Towards Learning-Augmented Peer-to-Peer Networks: Self-Stabilizing Graph Linearization with Untrusted Advice},
  year         = {2025},
  howpublished = {arXiv:2504.02448}
}

@inproceedings{balliuetal2025,
  author    = {Alkida Balliu and Sebastian Brandt and Fabian Kuhn and Krzysztof Nowicki and Dennis Olivetti and Eva Rotenberg and Jukka Suomela},
  title     = {Distributed Computation with Local Advice},
  booktitle = {39th International Symposium on Distributed Computing (DISC)},
  series    = {LIPIcs},
  volume    = {356},
  pages     = {12:1--12:19},
  year      = {2025}
}

\appendix

\section{Pseudo-code of the Predicted-Greedy Routing Algorithm}
\label{sec:alg}

\begin{algorithm}
\caption{Predicted-Greedy Routing (Cases 1 and 2)}
\label{alg:predicted-greedy}
\begin{algorithmic}[1]
\Require source $s$, target $t$
\State $U_0 \gets s$; $i \gets 0$
\While{$U_i \neq t$}
  \State receive prediction vector $\hat h_i$ (Case~1) or $\hat p_i$ (Case~2) for $N^+(U_i)$, conditioned on $\mathcal{F}_i$
  \State $U_{i+1} \gets \arg\min_{v \in N^+(U_i)} \hat h_i(v)$ (resp.\ $\hat p_i(v)$), ties broken by a fixed rule
  \State $i \gets i+1$
\EndWhile
\State \Return $T \gets i$
\end{algorithmic}
\end{algorithm}

\end{document}